\documentclass[letterpaper, 10 pt, conference]{ieeeconf}

\IEEEoverridecommandlockouts

\usepackage{mathptmx} 
 \usepackage{times} 
\usepackage[utf8]{inputenc}
\usepackage{amssymb}
\usepackage{amsmath}
\usepackage{graphicx}
\usepackage{hyperref}
\usepackage{bbm}
\usepackage{algorithm}
\usepackage[noend]{algpseudocode}
\usepackage{cite}
\usepackage{amsmath,amssymb,amsfonts,amsthm}
\usepackage{textcomp}
\usepackage{epstopdf}
\usepackage{tikz}
\usetikzlibrary{calc}
\usetikzlibrary{matrix, positioning, fit, calc}
\usepackage{tkz-euclide}
\usepackage{graphics} 
\usepackage{epsfig} 
 
\usepackage{cite}  
\usepackage{xcolor}
\usepackage{bbm}
\usepackage{enumerate}
\usepackage{algorithm}
\usepackage{amstext}
\usepackage{subfig}
\usepackage{multicol}

\theoremstyle{plain}

\newtheorem{assumption}{Assumption}

\newtheorem{proposition}{Proposition}
\newtheorem{remark}{Remark}

\newtheorem{lemma}{Lemma}

\newtheorem{theorem}{Theorem}

\newcommand{\G}{\mathcal{G}}
\newcommand{\V}{\mathcal{V}}

\newcommand{\E}{\mathcal{E}}

\newcommand{\N}{\mathcal{N}}

\renewcommand{\L}{\mathbf{L}}

\usepackage{accents}

\newcommand*{\dif}{\mathop{}\!\mathrm{d}}

\newcommand{\real}{\mathbb{R}}
\newcommand{\complex}{\mathbb{C}}

\newcommand{\naturals}{\mathbb{N}}
\newcommand{\F}{\mathbb{F}}

\newcommand{\lb}{[\![}
\newcommand{\rb}{]\!]}
\newcommand{\sys}[2]{\lb #1, #2\rb}
\newcommand{\innerp}[2]{\left \langle #1, #2\right \rangle}

\title{Structured stability analysis of networked systems with uncertain links
\thanks{Supported in part by the Australian Research Council (DP210103272).}}
\author{Simone Mariano and Michael Cantoni
	\thanks{S. Mariano and M. Cantoni are with the Department of Electrical and Electronic Engineering, The University of Melbourne, Australia. E-mails: {\tt \{simone.mariano,cantoni\}@unimelb.edu.au}}}

\begin{document}

\maketitle
\begin{abstract}
An input-output approach to stability analysis is explored for networked systems with uncertain link dynamics. The main result consists of a collection of integral quadratic constraints, which together imply robust stability of the uncertain networked system, under the assumption that stability is achieved with ideal links. The conditions are decentralized inasmuch as each involves only agent and uncertainty model parameters that are local to a corresponding link. This makes the main result, which imposes no restriction on network structure, suitable for the study of large-scale systems.
\end{abstract}

\section{Introduction}

Large-scale networks of dynamical systems arise in diverse contexts, including power and water distribution, transportation, manufacturing, ecology, and economics. Such networks have been studied in systems and control theoretic terms for many years; e.g., see~\cite{siljak1978large,arcak2016networks} for state-space methods, and~\cite{moylan1978stability,vidyasagar1981input} for input-output methods. 

An input-output approach is pursued in this paper. The aim is to devise scalable conditions for assessing the robustness of networked system stability to uncertainty in the link dynamics. This aim is achieved via application of the well-known integral quadratic constrainst (IQC) based robust stability theorem from~\cite{megretski1997system}, to a structured feedback interconnection model of the network, along lines related to the work reported in~\cite{lestas2006scalable,jonsson2010scalable,andersen2014robust,khong2014scalable,pates2016scalable}. Unlike~\cite{andersen2014robust,khong2014scalable}, scalability is achieved without particular restrictions on network structure. 

The developments here are most closely related to \cite{pates2016scalable}, where an approach to decomposing global IQC stability certificates is proposed for a structured feedback interconnection. However, the analysis underpinning~\cite[Theorem 1]{pates2016scalable} is not directly applicable to the model formulated here. This stems from the subsequent focus on link uncertainty relative to the ideal unity link, as is common in the study of cyber-physical systems~\cite{heemels2010networked,Cantoni2020}, and the use of a permutation `routing' matrix to encode network structure, as in~\cite{langbort2004distributed}. The difference is overcome by expanding upon ideas from \cite{pates2016scalable}, and leveraging a structured coprime factorization of the ideal network dynamics.   

The paper is organized as follows: Various preliminaries are established next. The structured feedback model of a networked system with uncertain links is given in Section~\ref{sec:ideal}. The IQC robust stability theorem is then applied in Section~\ref{sec:Robust_stability}, to arrive at a structured robust stability condition that is amenable to link-wise decomposition for decentralized verification, as explored in Section~\ref{sec:struct}. Some concluding remarks are provided in Section~\ref{sec:conc}.

\section{Preliminaries} 

\subsection{Basic notation} \label{subsec:notation}
The natural, real, and complex numbers are denoted by $\naturals$, $\real$, and $\complex$, respectively.
Given $i<j\in\naturals$, $[i:j]:=\{k\in\naturals~|~i\leq k \leq j\}$, and $\real_{\bullet \alpha}:=\{\beta\in\mathbb{R}~|~\beta \bullet \alpha\}$ for given order relation $\bullet\in\{>,\geq,<,\leq\}$. 

With $\F\in\{\real,\complex\}$, the $p$-dimensional Euclidean space over $\F$ is denoted by $\F^p$, for $p\in\naturals$. Given $x\in\F^p$, for $i\in[1:p]$, the scalar $x_i\in\F$ denotes the $i$-th coordinate.
The vector $\mathbf{1}_p\in\mathbb{R}^p$ is such that $(\mathbf{1}_p)_i=1$ for every $i\in[1:p]$.

$\F^{p\times q}$ denotes the space of $p\times q$ matrices over $\F$, for $p,q\in\naturals$.
The identity matrix is denoted by $I_p \in \mathbb{F}^{p\times p}$, the square zero matrix by $O_p\in \mathbb{F}^{p\times p}$, and the respective $p\times q$ matrices of ones and zeros by $\mathbf{1}_{p\times q}$ and $\mathbf{0}_{p\times q}$. Given $M\in \F^{p\times q}$, for $i\in[1:p]$, $j\in[1:q]$, the respective matrices $M_{(\cdot,j)} \in \F^{p\times 1}\sim\F^{p}$, and $M_{(i,\cdot)}\in\F^{1\times q}$, denote the $j$-th column, and $i$-th row. Further, $M_{(i,j)}\in\F$ denotes the entry in position $(i,j)$. Given $x\in\mathbb{F}^p$, the matrix $M=\mathrm{diag}(x)\in\mathbb{F}^{p\times p}$ is such that $M_{(i,i)}=x_i$ and $M_{(i,j)}=0$, $j\neq i\in[1:p]$, whereby $I_p=\mathrm{diag}(\mathbf{1}_p)$. The transpose of $M\in\mathbb{F}^{p\times q}$ is denoted by $M^\prime\in\mathbb{F}^{q\times p}$, and $M^*=\bar{M}^\prime$ denotes the complex conjugate transpose. For $M_i=M_i^*\in\mathbb{F}^{p\times p}$, $i\in\{1,2\}$, $M_i\succ 0$ means there exists $\epsilon>0$ such that $x^* M_i x \geq \epsilon x^*x$ for all $x\in\mathbb{F}^p$, $M_i\succeq 0$ means $x^* M_i x \geq 0$ for all $x\in\mathbb{F}^p$, and $M_1\succeq (\text{resp.}~\succ) M_2$ means $M_1-M_2 \succeq (\text{resp.}~\succ) 0$. 

\subsection{Signals and systems}
The Hilbert space of square integrable signals $v=(t\in\real_{\geq 0}\mapsto v(t)\in\real^p)$ is denoted by $\L_{2\,}^p$, where the inner-product $\langle v, u \rangle := \int_{0}^{\infty}v(t)^\prime u(t)\dif t$ and norm $\|v\|_2 := \langle v, v \rangle^{1/2}$ are finite; the superscript is dropped when $p=1$. The corresponding extended space of locally square integrable signals is denoted by $\L_{2e}^p$; i.e.,  $v:\real_{\geq 0}\rightarrow\real^p$ such that $\boldsymbol{\pi}_{\tau}(v)\in\L_{2\,}^p$ for all $\tau\in\real_{\geq 0}$, where $(\boldsymbol{\pi}_\tau(v))(t):=f(t)$ for $t\in[0,\tau)$, and $(\boldsymbol{\pi}_\tau(v))(t):=0$ otherwise. 
The composition of maps $F:\L_{2e}^p \mapsto \L_{2e}^r$ and $G:\L_{2e}^q \mapsto \L_{2e}^p$ is denoted by $F\circ G := (v\mapsto F(G(v))$, and the direct sum by $F\oplus G:=((u,v)\mapsto (F(u),G(v)))$. Similarly, $\bigoplus_{i=1}^n G_i = G_1\oplus\cdots\oplus G_n$. When $G$ is linear, in the sense $(\forall \alpha,\beta\in\real)~(\forall u,v\in\L_{2e}^q)~G(\alpha u + \beta v) = \alpha G(u) + \beta G(v)$, the image of $v$ under $G$ is often written $Gv$, and in composition with another linear map the $\circ$ dropped. The action of a linear system $G:\L_{2e}^{q}\rightarrow\L_{2e}^{p}$ corresponds to the action of $p\cdot q$ scalar systems $G_{(i,j)}:\L_{2e}\rightarrow \L_{2e}$, $i\in[1:p]$, $j\in[1:q]$, on the coordinates of the signal vector input associated with  $\L_{2e}^q\sim\L_{2e}\times\cdots\times\L_{2e}$; i.e., $(Gv)_i = \sum_{j=1}^q G_{(i,j)}v_j$. Matrix notation is used to denote this.

A system is any map $G:\L_{2e}^q\rightarrow \L_{2e}^p$, with $G(0)=0$, that is {\em causal} in the sense  $\boldsymbol{\pi}_\tau(G(u)) = \boldsymbol{\pi}_\tau(G (\boldsymbol{\pi}_\tau(u))$ for all $\tau\in\real_{\geq 0}$. It is
called {\em stable} if $u\in\L_{2\,}^q$ implies $G(u) \in \L_{2\,}^p$ and  $\|G\|:=\sup_{0\neq u} \|G(u)\|_2/\|u\|_2< \infty$. The composition of stable systems is therefore stable. The feedback interconnection with system $\Delta:\L_{2e}^p \rightarrow \L_{2e}^q$ is {\em well-posed} if for all $(d_y,d_u)\in\L_{2e}^p\times\L_{2e}^q$, there exists unique $(y,u)\in\L_{2e}^p\times\L_{2e}^q$, such that 
\begin{align} \label{eq:feedback1}
y = G(u) + d_y, \qquad u = \Delta(y) + d_u, 
\end{align}
and $[\![G,\Delta]\!] := ((d_y,d_u) \mapsto (y,u))$ is causal; see Figure~\ref{fig:realnet_0}. If, in addition, $\|[\![G,\Delta]\!]\| < \infty$, then the closed-loop is called stable. 
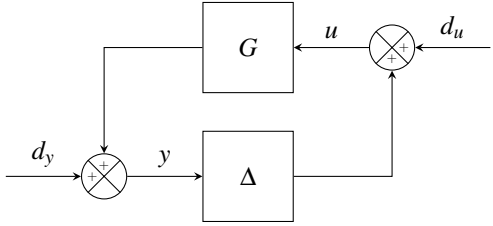
\begin{figure}[htbp]
\centering
\begin{tikzpicture}
 
\node [draw,
    minimum width=1.2cm,
    minimum height=1.2cm,
]  (permutation) at (0,0) {$G$};
 
\node [draw,
    minimum width=1.2cm, 
    minimum height=1.2cm, 
    below=.5cm  of permutation
]  (nets) {$\Delta$};

\node[draw,
    circle,
    minimum size=0.6cm,
    left=1cm  of nets
] (sum){};
 
\draw (sum.north east) -- (sum.south west)
(sum.north west) -- (sum.south east);
 
\node[left=-1pt] at (sum.center){\tiny $+$};
\node[above=-1pt] at (sum.center){\tiny $+$};


\node[draw,
    circle,
    minimum size=0.6cm,
  right= 1cm   of permutation 
] (sum2){};
 
\draw (sum2.north east) -- (sum2.south west)
(sum2.north west) -- (sum2.south east);
 
\node[right=-1pt] at (sum2.center){\tiny $+$};
\node[below=-1pt] at (sum2.center){\tiny $+$};

\draw[-stealth] (sum.east) -- (nets.west)
    node[midway,above]{$y$};
 
\draw[-stealth] (permutation.west) -| (sum.north);
 
\draw [stealth-] (sum.west) -- ++(-1,0) 
    node[midway,above]{$d_{y}$};

\draw[-stealth]  (nets.east) -|  (sum2.south);
 
\draw[-stealth] (sum2.west) -- (permutation.east) 
    node[midway,above]{$u$};
 
\draw [stealth-] (sum2.east) -- ++(1,0) 
    node[midway,above]{$d_{u}$};

\end{tikzpicture}
    \caption{Standard feedback interconnection. }
    \label{fig:realnet_0}      
\end{figure}

\begin{lemma} \label{lem:stabstab}
 If $G:\L_{2e}^q\rightarrow \L_{2e}^p$ and $\Delta:\L_{2e}^p\rightarrow \L_{2e}^q$ are both stable, and $G$ is linear, then the following are equivalent:
 \begin{enumerate}[i)]
     \item $[\![G,\Delta]\!]$ is stable; \label{it:stable}
     \item $(I_p-G\circ\Delta):\L_{2e}^p\rightarrow\L_{2e}^p$ has a causal inverse that is also stable, where pointwise multiplication by $I_p$ is not distinguished from the matrix for convenience. \label{it:causalinv} 
 \end{enumerate}
\end{lemma}
\begin{proof}
See Appendix.
\end{proof}

The following result is a well-known robust feedback stability theorem, taken from~\cite{megretski1997system}:
\begin{theorem}
\label{thm:robust_stability}
Given stable system $\Delta:\L_{2e}^p \rightarrow \L_{2e}^q$, and bounded linear self-adjoint $\Pi:\L_{2\,}^p\times\L_{2\,}^q \rightarrow \L_{2\,}^p\times \L_{2\,}^q$, in the sense $\langle g_1, \Pi g_2 \rangle=\langle \Pi g_1, g_2\rangle$ for all $g_1,g_2\in\L_{2\,}^p\times \L_{2\,}^q$, suppose 
\begin{equation}
\label{eq:stab_IQC_Delta}
\left\langle 
(y,u),
\Pi 
(y,u)
\right\rangle \geq 0,
\end{equation}
with $u=\alpha \Delta(y)$, for all $y\in\L_{2\,}^p$, and $\alpha\in[0,1]$. Further, given stable $G:\L_{2e}^q\mapsto \L_{2e}^p$, suppose $[\![G,\alpha \Delta]\!]$ is well-posed for all $\alpha\in[0,1]$, and there exists $\epsilon>0$ such that 
\begin{equation}
\label{eq:stab_G}
\left\langle 
(y,u),
\Pi 
(y,u)
\right\rangle \leq -\epsilon \| u\|_2^2,
\end{equation}
with $y = G(u)$,
for all $u\in\L_{2\,}^q$.
 Then, $\sys{G}{\Delta}$ is stable. 
\end{theorem}
\begin{remark}
    The complementary constraints in \eqref{eq:stab_IQC_Delta} and \eqref{eq:stab_G} are IQCs. 
\end{remark}

\subsection{Graphs}
Let $\G=(\V,\E)$ be a simple (self-loopless and undirected) graph, where $\V=[1:n]$ is the set of $n\in\mathbb{N}$ vertices, and $\E\subset\{\{i,j\}~|~i,j\in\V\}$ is the set of $m=|\E|\in\mathbb{N}$ edges. The set $\N_i:=\{j~|~\{i,j\}\in\E\}$ comprises the neighbours of $i\in\V$. Further, $\E_i:=\{\{i,j\}~|~j\in\mathcal{N}_i\}$ is the neighbourhood edge set, $\G_i:=\G[\E_i]$ is the $\E_i$-induced sub-graph of $\G$, and bijective $\kappa_{\E_i}:\E_i\rightarrow [1:m_i]$ denotes a fixed enumeration of $\E_i$. Similarly, bijective $\kappa_{\E}:\E\rightarrow [1:m]$ denotes a fixed enumeration of the edge set $\E$.
%
%

\section{Networked system model}
\label{sec:ideal}

Consider a network of $n\in \naturals$ dynamic agents, coupled according to the simple graph $\G=(\V,\E)$. The vertex set $\V:=[1:n]$ corresponds to a fixed enumeration of the agents, and  
$m:=|\E|$ is the number of edges, defined such that $\{i,j\}\in\E$ if the output of agent $i\in\V$ is shared as an input to agent $j\in\V$, and vice-versa; see Figure~\ref{fig:example_net}. It is assumed that the number of neighbours $m_i:=|\N_i|\geq 1$ for all $i\in\V$. To tame the notation, each agent has a single output, and multiple inputs, one for each neighbour. The corresponding input-output system  $H_i:\L_{2e}^{m_i}\rightarrow \L_{2e}$ is taken to be linear and stable, with vector input signal coordinate order fixed by the neighbourhood edge-set enumerations $\kappa_{\E_i}$.

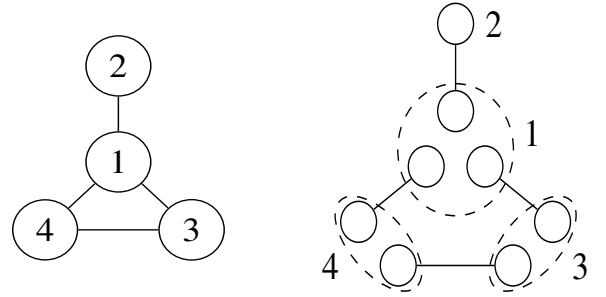
\begin{figure}[htbp]
\begin{minipage}{.48\linewidth}
\centering
\resizebox{3cm}{3cm}{
\begin{tikzpicture}[main/.style = {draw, circle}] 
\node[main] (1) {1}; 
\node[main] (2) [above of=1] {2};
\node[main] (3) [below right of=1] {3}; 
\node[main] (4) [below left of=1] {4};
\draw (1) -- (2);
\draw (1) -- (3);
\draw (1) -- (4);
\draw (4) -- (3);
\end{tikzpicture}} 
\end{minipage}
\hfill
\begin{minipage}{.48\linewidth}
\centering
\resizebox{4cm}{4cm}{
\begin{tikzpicture}[main/.style = {draw, circle}] 

\node[main] (1) {\,}; 
\node[main] (2) [below of=1,yshift=2mm] {\,};
\node[main] (3) [below right of=2,yshift=2mm,xshift=-4mm] {\,}; 
\node[main] (4) [below left of=2,yshift=2mm,xshift=4mm] {\,};
\node[main] (5) [below left of=4, yshift=+2mm] {\,};
\node[main] (6) [below right of=4, yshift=-2mm,xshift=-10mm] {\,};
\node[main] (7) [below left of=3,yshift=-2mm,xshift=10mm] {\,};
\node[main] (8) [below right of=3,yshift=+2mm] {\,};

\node [rotate=43][draw,dashed,inner sep=0pt, circle,yscale=.5, fit={(7) (8)}] {};
\node [draw,dashed,inner sep=0pt, circle,yshift=-1mm,xscale=.9,yscale=.9, fit={(2) (3) (4)}] {};
\node[rotate=-43][draw,dashed,inner sep=0pt, circle,yscale=.5, fit={(5) (6)}] {};

\node [below of=1, xshift=8mm] {1};
\node [right of=1, xshift=-6mm] {2};
\node [below left  of=8, yshift=3mm, xshift=10mm] {3};
\node [below right of=5, yshift=3mm, xshift=-10mm] {4};

\draw (1) -- (2);
\draw (4) -- (5);
\draw (3) -- (8);
\draw (6) -- (7);

\end{tikzpicture}}
\end{minipage}
\caption{Example network graph $\G$ (left) and corresponding sub-system graph $\G_s$ (right). The sub-system graph has $2m=8$ vertices and $m=4$ edges, while the network graph has $n=4$ nodes and $m=4$ edges.}  
\label{fig:example_net}                    
\end{figure} 

The network of agents can be modelled as the feedback interconnection of structured open-loop systems. To this end, noting that $\sum_{i=1}^n m_i=2m$, define the block diagonal systems
\begin{subequations}
\label{eq:blkdiag}
\begin{align}
H&:=\bigoplus_{i=1}^{n}H_i:\L_{2e}^{2m}\rightarrow \L_{2e}^{n}, \\  
T&:=\bigoplus_{i=1}^{n} \mathbf{1}_{m_{i}\times 1}:\L_{2e}^{n}\rightarrow \L_{2e}^{2m}, \\ 
\text{ and } \quad R&:=\bigoplus_{i=1}^{n} \big(\bigoplus_{k=1}^{m_i} R_{i,k}\big):\L_{2e}^{2m}\rightarrow\L_{2e}^{2m},
\end{align}
\end{subequations}
 where $\mathbf{1}_{m_{i}\times 1}:\L_{2e}\rightarrow\L_{2e}^{m_i}$ denotes pointwise multiplication by the matrix $\mathbf{1}_{m_{i}\times 1}$, and the system $R_{i,k}:\L_{2e}\rightarrow \L_{2e}$ represents the stable but possibly nonlinear link dynamics from agent $i\in\V$ to the neighbour $j\in\N_i$ for which $\{i,j\}=\kappa_{\E_i}^{-1}(k)$. For an {\em ideal} link $R_{i,k}$, is pointwise multiplication by the scalar $1$. Given these components, the networked system corresponds to the structured feedback interconnection $[\![P, R \circ T \circ H]\!]$ shown in the top left of Figure~\ref{fig:transfloop}, where $P:\L_{2e}^{2m}\rightarrow \L_{2e}^{2m}$ is pointwise multiplication by the permutation matrix that `routes' the link outputs to the agent inputs in accordance with $\G$, and the fixed neighbourhood edge-set enumerations $\kappa_{\E_i}$. More specifically, for each $i\in\V$, $k\in[1:m_i]$, and $r\in[1:2m]$, the corresponding entry of this permutation matrix is given by
 \begin{align} \label{eq:permute}
 P_{(\sum_{h=1}^{i} m_h-m_i+k,\,r)}\!=\!
 \begin{cases} 1 & \text{if}~\!r\!=\!
 \sum_{h=1}^j \!m_h \!-\! m_j
 \!+\!\kappa_{\E_j}(\!\{i,j\}\!)
 \\ &\text{with } j\in\kappa_{\E_i}^{-1}(k)\setminus \{i\},\\
 0 & \text{otherwise}.
 \end{cases}
 \end{align}
 
\begin{figure}[tbp]
\hspace{0pt} 
\begin{minipage}{.48\linewidth}
\centering
\hspace{-5pt} 
\resizebox{4cm}{2cm}{\begin{tikzpicture}
 
\node [draw,
    minimum width=1.2cm,
    minimum height=1.2cm,
]  (permutation) at (0,0) {$P$};
 
\node [draw,
    minimum width=2cm, 
    minimum height=1.2cm, 
    below=.5cm  of permutation
]  (nets) {$R\circ T \circ H$};

\node[draw,
    circle,
    minimum size=0.6cm,
    left=0.8cm  of nets
] (sum){};
 
\draw (sum.north east) -- (sum.south west)
(sum.north west) -- (sum.south east);
 
\node[left=-1pt] at (sum.center){\tiny $+$};
\node[above=-1pt] at (sum.center){\tiny $+$};


\node[draw,
    circle,
    minimum size=0.6cm,
  right= 1.2cm   of permutation 
] (sum2){};
 
\draw (sum2.north east) -- (sum2.south west)
(sum2.north west) -- (sum2.south east);
 
\node[right=-1pt] at (sum2.center){\tiny $+$};
\node[below=-1pt] at (sum2.center){\tiny $+$};

\draw[-stealth] (sum.east) -- (nets.west)
    node[midway,above]{$v$};
 
\draw[-stealth] (permutation.west) -| (sum.north);
 
\draw [stealth-] (sum.west) -- ++(-0.5,0) 
    node[midway,above]{$d_{v}$};

\draw[-stealth]  (nets.east) -|  (sum2.south);
 
\draw[-stealth] (sum2.west) -- (permutation.east) 
    node[midway,above]{$w$};
 
\draw [stealth-] (sum2.east) -- ++(0.5,0) 
    node[midway,above]{$d_{w}$}; 
\end{tikzpicture}} 
\! \vspace{5 pt} \, 
\resizebox{4.5cm}{2.25cm}{\begin{tikzpicture}
 
\node [draw,
    minimum width=1.4cm,
    minimum height=1.2cm,
]  (permutation) at (0,0) {$H \circ P$};
 
\node [draw,
    minimum width=1.4cm, 
    minimum height=1.2cm, 
    below=.5cm  of permutation
]  (nets) {$R\circ T$};

\node[left=1cm  of nets] (sum){};


\node[draw,
    circle,
    minimum size=0.6cm,
  right= 1cm   of permutation 
] (sum2){};
 
\draw (sum2.north east) -- (sum2.south west)
(sum2.north west) -- (sum2.south east);
 
\node[right=-1pt] at (sum2.center){\tiny $+$};
\node[below=-1pt] at (sum2.center){\tiny $+$};

\draw[-stealth] (sum.center) -- (nets.west)
    node[midway,above]{$\tilde{v}$};
 
\draw(permutation.west) -| (sum.center);

\draw[-stealth]  (nets.east) -|  (sum2.south);
 
\draw[-stealth] (sum2.west) -- (permutation.east) 
    node[midway,above]{$\tilde{w}$};
 
\draw [stealth-] (sum2.east) -- ++(0.5,0) 
    node[midway,above]{$\qquad \qquad d_w+P^{-1}d_v$};
\end{tikzpicture}} 
\end{minipage}
\hspace{3pt} 
\begin{minipage}{.46\linewidth}
\centering
\resizebox{4.2cm}{5cm}{\begin{tikzpicture}
 
\node [draw,
    minimum width=1.4cm,
    minimum height=1.2cm,
]  (permutation) at (0,0) {$H \circ P$};
 
\node [draw,
    minimum width=1.4cm, 
    minimum height=1.2cm, 
    below=4.9cm  of permutation
]  (nets) {$R\circ T$};

\node [draw,
    minimum width=1.0cm, 
    minimum height=1.0cm, 
    below=.3cm  of permutation
]  (T1) {$T$};

\node [draw,
    minimum width=1.0cm, 
    minimum height=1.0cm, 
    above=.3cm  of nets
]  (T2) {$T$};

\node[left=1cm  of nets] (n1){};

\node[above=1cm  of nets] (n2){};

\node[left=1cm  of n2] (n3){};

\node[left=1cm  of permutation] (n4){};

\node[right=1cm  of permutation] (n5){};

\node[left=1.2cm  of T1] (n6){};

\node[left=1.2cm  of T2] (n7){};

\node[right=1.2cm  of T1] (n8){};

\node[right=1.2cm  of T2] (n9){};

\node[left=1cm  of nets] (n10){};

\node [ 
    above=0.8cm  of n10
]  (S) {};

\node [
    below=0.8cm  of n4
]  (Sda) {};


\node[draw,
    circle,
    minimum size=0.6cm,
  below=2.7cm of n5
] (sum2){};
 
\draw (sum2.north east) -- (sum2.south west)
(sum2.north west) -- (sum2.south east);
 
\node[right=-1pt] at (sum2.center){\tiny $+$};
\node[below=-1pt] at (sum2.center){\tiny $+$};


\node[draw,
    circle,
    minimum size=0.6cm,
  right=1.0cm of T2
] (sum3){};
 
\draw (sum3.north east) -- (sum3.south west)
(sum3.north west) -- (sum3.south east);
 
\node[left=-1pt] at (sum3.center){\tiny $-$};
\node[below=-1pt] at (sum3.center){\tiny $+$};


\node[draw,
    circle,
    minimum size=0.6cm,
  right=1.0cm of T1
] (sum4){};
 
\draw (sum4.north east) -- (sum4.south west)
(sum4.north west) -- (sum4.south east);
 
\node[left=-1pt] at (sum4.center){\tiny $+$};
\node[below=-1pt] at (sum4.center){\tiny $+$};

\draw[-stealth] (n1.center) -- (nets.west)
    node[midway,above]{$\tilde{v}$};
 
\draw(permutation.west) -| (Sda.center);

\draw(Sda.center) -| (S.center);

\draw(S.center) -- (n1.center);

\node[above=1.6cm of S] (ttS) {}; 
\node[left=-0.1 of ttS] {$y$};
\node[above=0.2cm of sum2] (ttsum) {};
\node[right=-0.1cm of ttsum] {$u$};

\draw[-stealth]  (nets.east) -|  (sum3.south);
 
\draw[-stealth] (n5.center) -- (permutation.east) 
    node[midway,above]{$\tilde{w}$};
    
\draw[-stealth] (n6.center) -- (T1.west);
\draw[-stealth] (n7.center) -- (T2.west);

\draw[stealth-] (sum4.west) -- (T1.east);
\draw[stealth-] (sum3.west) -- (T2.east);

\draw[-stealth] (sum2.north) -- (sum4.south);
 
\draw[-stealth] (sum3.north) -- (sum2.south);

\draw (n5.center) -- (sum4.north) 
    node[midway,above]{};
 
\draw [stealth-] (sum2.east) -- ++(0.5,0) 
    node[right,above]{$d$};

 \node[draw,dashed,inner sep=7pt, yscale=1, fit={(T1) (Sda) (permutation) (sum4)}] (G) {};

 \node[draw,dashed,inner sep=7pt, yscale=1, fit={(T2) (S) (nets) (sum3)}] (D) {};

\node[left=0.25cm of D] (D1) {$\Delta$};
\node[left=0.25cm of G] (G1) {$G$};
\node[midway,above]{$\hspace{-2.5 cm}{\vspace{0.2cm}\tilde{v}}$};
 
 
\end{tikzpicture}}
\end{minipage}
    \caption{Networked system model $[\![P,R\circ T\circ H]\!]$, and loop transformations for robust stability analysis.}
    %
    \label{fig:transfloop}      
\end{figure}
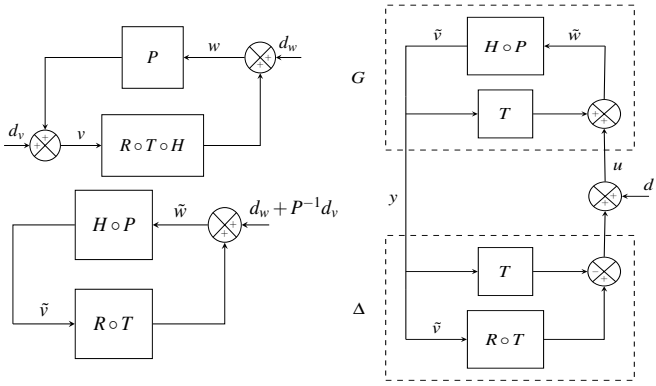
 
The matrix $P=P^\prime=P^{-1}$ is the adjacency matrix of an undirected $1$-regular sub-system graph 
 $\widetilde{\G}:=(\widetilde{\V},\widetilde{\E})$, with $2m$ vertices and $m$ edges. 
 More specifically, the disjoint union $\bigsqcup_{e\in\E}\G[e]$ of the two-vertex sub-graphs induced by each edge of the network graph $\G=(\V,\E)$; see Figure~\ref{fig:example_net}. As such, 
\begin{equation}
 \label{eq:Padjacency}
 P=D-L=I_{2m}-
 \sum_{k=1}^{m} L_k,
 \end{equation}
where the degree matrix $D=I_{2m}$ because $\widetilde{\G}$ is 1-regular, and the graph Laplacian decomposes as
 \begin{equation}
 \label{eq:Laplacian_Subsystem_Graph}
  L=\sum_{k=1}^{m}L_k
  =\sum_{k=1}^{m}B_{(\cdot,k)}B_{(\cdot,k)}^\prime
  =B B^\prime,   
 \end{equation}
 with incidence matrix $B\in\mathbb{R}^{2m\times m}$ 
 defined by 
$B_{(r,k)}:=1$, $B_{(s,k)}:=-1$, and $B_{(l,k)}:=0$ for $l \in [1:2m]\setminus\{r,s\}$, with $\{r,s\}=\kappa^{-1}_{\widetilde{\E}}(k)$, for $k\in[1:m]$. The enumeration $\kappa_{\widetilde{\E}}$ is taken to be compatible with the enumerations $\kappa_{\E_i}$, and the definition of $P$ in \eqref{eq:permute}. The edge orientation is arbitrary. 
 
 As above, for convenience throughout, the system corresponding to pointwise multiplication by a matrix is not distinguished from the matrix. 
Since $P=P^{-1}:\L_{2e}\rightarrow\L_{2e}$ are linear and stable systems, 
stability of $[\![P, R\circ T\circ H]\!]$ is equivalent to stable invertibility of $(I_{2m}-(R\circ T)\circ(H\circ P))$, where the parenthetic grouping separates the link dynamics from the agent dynamics. The equivalence, illustrated in the left-hand side of Figure~\ref{fig:transfloop}, is formally established below.
\begin{lemma} \label{lem:equiv1}
   $[\![P, R\circ T\circ H]\!]$ is stable if, and only if, the stable system
   $(I_{2m}-(R\circ T)\circ(H\circ P))$ has a stable inverse.
\end{lemma}
\begin{proof}
Since $R\circ T\circ H: \L_{2e}^{2m}\rightarrow \L_{2e}^{2m}$ and $P:\L_{2e}\rightarrow\L_{2e}$ are both stable by hypothesis, and $P$ is linear, Lemma~\ref{lem:stabstab} applies to give $[\![P, R\circ T\circ H]\!]$ is stable if, and only if, $(I-P\circ(R\circ T \circ H))$ has a stable inverse. Further, since $P^{-1}=P$, it follows that $(I-P\circ(R\circ T \circ H)) =P\circ (I-(R\circ T)\circ(H\circ P)) \circ P^{-1}$, whereby stable invertibility of $(I-P\circ(R\circ T \circ H))$ is equivalent to stable invertibility of $(I-(R\circ T)\circ(H\circ P))$.
\end{proof}
 
The focus of subsequent developments is on the robustness of network stability to link uncertainty specified with respect to the ideal unity link. To this end, define 
\begin{equation}
\label{eq:Delta}
    \Delta:= (R-I_{2m}) \circ T,
\end{equation}
with the stable systems $R$ and $T$ as per \eqref{eq:blkdiag}, whereby $\Delta:\L_{2e}^{n}\rightarrow \L_{2e}^{2m}$ is stable. Then, provided $(P - T \circ H )$ has a stable inverse, with $H$ and $P$ as per \eqref{eq:blkdiag} and \eqref{eq:permute}, stable invertibility of $(I-(R\circ T)\circ(H\circ P))$, and therefore, stable network dynamics, is implied by stability of the feedback interconnection $[\![G,\Delta]\!]$, where 
\begin{equation}
\label{eq:G}
G := 
H \circ (P - T\circ H )^{-1}.
\end{equation}
This implication, illustrated on the right-hand side of Figure~\ref{fig:transfloop}, is established more formally below.
\begin{assumption}
\label{ass:coprime} 
 The network with ideal links is stable in the sense that $[\![P, T\circ H]\!]$ 
 is stable.
\end{assumption}
\begin{remark} \label{rem:nomstab}
     In view of Lemma~\ref{lem:stabstab}, stability of the ideal network $[\![P, T\circ H]\!]$ is equivalent to stable invertibility of $(I_{2m}-P\circ T\circ H)$, and thus, stable invertibility of $(P-T\circ H) =  P\circ( I_{2m}-P\circ T\circ H )$, by the linearity of $P^{-1}=P$. So under Assumption~\ref{ass:coprime}, $G$ in \eqref{eq:G} is stable.
     \end{remark}
\begin{theorem}
\label{thm:stable_sampledH}
Under Assumption~\ref{ass:coprime}, if 
$[\![G,\Delta]\!]$ is stable, with $G$ as per \eqref{eq:G}, and $\Delta$ as per \eqref{eq:Delta}, then the networked system model $[\![P,R\circ T \circ H]\!]$ is stable. Further, when $\Delta$ is also linear, stability of $[\![P,R\circ T \circ H]\!]$ implies stability of $[\![G,\Delta]\!]$.  
\end{theorem}
\begin{proof}
By Lemma~\ref{lem:equiv1}, the networked system $[\![P,R\circ T \circ H]\!]$ is stable if, and only if, $(I_{2m} - (R\circ T)\circ(H\circ P))$ has a stable inverse. 

With reference to Figure~\ref{fig:transfloop}, suppose $[\![G,\Delta]\!]$ is stable, so that for every $d\in\L_{2e}^{2m}$, there exists unique $(y,u)\in\L_{2e}^n \times \L_{2e}^{2m}$, with causal dependence on $d$, such that $u = (\Delta\circ G)(u) + d = \Delta(y) + d = -Ty + (R\circ T)(y) + d$, and $y = Gu = H\circ(P - T\circ H)^{-1}u$. Further, $\|u\|_2 \leq c_u \|d\|_2$, and $\|y\|_2 \leq c_y \|d\|_2$ when $d\in\L_{2\,}^{2m}$, for some fixed $c_u,c_y\in\mathbb{R}_{>0}$. In particular, $u \mapsto d$ is a bijection $\L_{2e}^{2m}\rightarrow\L_{2e}^{2m}$, and the inverse $d \mapsto u$ is stable. Let $\tilde{w} := (I_{2m} - T \circ H\circ P)^{-1}u = P \circ (P-T\circ H)^{-1} u$. In view of Assumption~\ref{ass:coprime}, and Remark~\ref{rem:nomstab}, $\tilde{w}\in\L_{2\,}^{2m}$, and $\|\tilde{w}\|_2 \leq c_w \|u\|_2 \leq c_w\cdot c_u \|d\|_2$ when $d\in\L_{2\,}^{2m}$, for some fixed $c_w\in\mathbb{R}_{>0}$. Observe that $\tilde{w}$ depends causally on $u$, and thus, on $d$. Further, $(H\circ P)\tilde{w} = Gu= y$, since $P^{-1}=P$, and $u = \tilde{w} - (T\circ H\circ P) \tilde{w} = \tilde{w} - Ty$. As such, $\tilde{w}-Ty = u= -Ty + ((R\circ T)\circ(H\circ P))(\tilde{w}) + d$, and thus, $(I-(R\circ T)\circ(H\circ P))(\tilde{w}) = d$. So $(d \mapsto \tilde{w})$ is a stable right inverse of $(I-(R\circ T)\circ(H\circ P)) = (\tilde{w}\mapsto d)$. That it is also the left inverse follows from the fact that $(I-(R\circ T)\circ(H\circ P))$ is injective. This can be seen by contradiction. Suppose there exist $\breve{w} \neq \tilde{w} 
\in\L_{2e}^{2m}$ such that 
$(I-(R\circ T)\circ(H\circ P))(\breve{w})=(I-(R\circ T)\circ(H\circ P))(\tilde{w})$, then $d:=\breve{w} - (R\circ T)(\breve{y}) = \tilde{w} - (R\circ T)(\tilde{y})$ with $\breve{y}=(H\circ P)\breve{w}$, and $\tilde{y}=(H\circ P)\tilde{w}$. Further, $\breve{u}=\breve{w}-T\breve{y}$ and $\tilde{u}=\tilde{w}-T\tilde{y}$ are such that $G\breve{u}=\breve{y}$, $G\tilde{u}=\tilde{y}$, $\breve{u}=\Delta(\breve{y})+d$, and $\tilde{u}=\Delta(\tilde{y})+d$, whereby $\breve{u}=\tilde{u}$ and $\breve{y}=\tilde{y}$ since $[\![G,\Delta]\!]$ is well-posed. Therefore, $\breve{w}-\tilde{w}=(R\circ T)(\breve{y})-(R\circ T)(\tilde{y})=0$, which is a contradiction. 

Similarly, if $(I-(R\circ T)\circ(H\circ P))$ has a stable inverse, then for every $d\in\L_{2e}^{2m}$, there exists unique $\tilde{w}= ((R\circ T)\circ(H\circ P))(\tilde{w}) + d =(R\circ T)(\tilde{v}) + d$, where $\tilde{v} = (H\circ P)\tilde{w}$. The map $d \mapsto \tilde{w}$ is causal, and $\|\tilde{w}\|_2 \leq c \|d\|_2$ when $d\in\L_{2\,}^{2m}$, for some fixed $c\in\mathbb{R}_{>0}$. As above, $u=\tilde{w}-Ty$, with $y=(H\circ P)\tilde{w}$, is such that $y=Gu$, and $u=\Delta(y)+d = (\Delta\circ G)(u)+d$. So $(d\mapsto u)$ is a right inverse of $(I_{2m}-\Delta\circ G)$. It is also the left inverse, since it can be established that $(I_{2m}-\Delta\circ G)$ is injective, for otherwise the uniqueness of $\tilde{w}$ can be contradicted along similar lines to the related argument above. Therefore, $(I_{2m}-\Delta\circ G)$ has a stable inverse, which implies stability of $[\![G,\Delta]\!]$ by Lemma~\ref{lem:stabstab}, if $\Delta$ is also linear.
\end{proof}

\begin{remark} 
\label{rem:coprime}
Let 
\begin{gather}
    \label{eq:coprimeNM}
    N:= H, \quad \mathrm{and} \quad    M:=(P-T\circ H),
\end{gather}
with $T$, $H$, and $P$, as in \eqref{eq:blkdiag} and \eqref{eq:permute}. Under Assumption~\ref{ass:coprime}, 
$(N,~M)$ is
a right coprime factor pair for the stable system 
$G=
H\circ (P-T\circ H)^{-1}$ in \eqref{eq:G}. In particular,
$N$ and $M$ are stable, 
$M^{-1}$ is stable,
and 
$\mathbf{0}_{2m\times n}\circ N + M^{-1}\circ M = I_{2m}$; see~\cite{vidyasagar1985control} for more on coprime factorizations and feedback. Importantly, these coprime factor pairs are structured, and $y = Gu $ with $(y,u)\in \L_{2\,}^n\times \L_{2\,}^{2m}$ if, and only if, there exists $z\in\L_{2\,}$ such that 
$u = Mz$, and $y = Nz$.
These features are exploited subsequently to enable structured application of Theorem~\ref{thm:robust_stability}, despite the potential lack of structure in $G$ due to the factor $(P-T\circ H)^{-1}$, which may become unstructured even though $P$, and $T\circ H$, are structured. 
\end{remark}

\section{Robust stability analysis}
\label{sec:Robust_stability}
Theorem~\ref{thm:robust_stability} is applied below to asses the stability of the uncertain network model $\sys{P}{R\circ T \circ H}$. This is achieved via Theorem~\ref{thm:stable_sampledH}, and analysis of $\sys{G}{\Delta}$ under Assumption~\ref{ass:coprime} regarding stability of the network with ideal links, whereby $G=H\circ(P-T\circ H)^{-1}$ in \eqref{eq:G} is stable. In particular, the developments lead to a collection of conditions that are together sufficient for stability of $\sys{G}{\Delta}$. Decentralized verification of this stability certificate is explored in Section~\ref{sec:struct}. There it is shown how to formulate collectively sufficient conditions, each involving only model data that is local to a corresponding link, and to a component of the uncertainty description, as per the stable system $\Delta=(R-I_{2m})\circ T=\oplus_{i=1}^n\Delta_i$ in \eqref{eq:Delta}, where 
$$\Delta_i = (\oplus_{k=1}^{m_i}R_{i,k} - I_{m_i})\circ \mathbf{1}_{m_i\times 1},\quad i\in\V=[1:n].$$
Recall that $m=|\E|$, and $m_i=|\N_i|$, where $\G=(\V,\E)$ is the network graph, and $\N_i$ is the set of agent $i\in\V$ neighbours.

The approach is underpinned by IQC characterization of the link uncertainty $\Delta$; i.e., the existence of a bounded linear self-adjoint operator $\Phi:(\L_{2\,}^{n} \times \L_{2\,}^{2m}) \rightarrow (\L_{2\,}^{n} \times \L_{2\,}^{2m})$ such that for all $\alpha \in [0,1]$, and $y \in \L_{2\,}^{n}$,
\begin{equation} \label{eq:DeltaIQC}
\innerp{(y,u)}{\Phi(y,u)} \geq 0,
\end{equation}
with $u=\alpha \Delta(y)$. Such $\Phi$ can be constructed from corresponding IQCs $\innerp{(y_i,u_i)}{\Phi_i (y_i,u_i)}\geq 0$ for all $\alpha\in[0,1]$, and $y_i\in\L_{2\,}$, with $u_i=\alpha\Delta_i(y_i)$, 
given bounded linear self-adjoint
$$\Phi_i := 
\begin{bmatrix} \Phi_{1,i} & \Phi_{2,i} \\ \Phi_{2,i}^*
 & \Phi_{3,i} \end{bmatrix}
:\L_{2\,}\times\L_{2\,}^{m_i} \rightarrow 
\L_{2\,} \times \L_{2\,}^{m_i},\quad i\in\V=[1:n];$$
where the superscript $*$ here denotes the Hilbert adjoint, such that $\langle y,\Phi_{2,i} u\rangle = \langle\Phi_{2,i}^* y,u\rangle$ for all $(y,u)\in\L_{2\,}^{n}\times \L_{2\,}^{2m}$ (it exists for all bounded linear operators between Hilbert spaces like $\L_{2\,}$~\cite{Kreyszig}.) Given such IQC characterizations of the uncertain links, the bounded linear self-adjoint operator
\begin{align*}
    \Phi = \begin{bmatrix} \Phi_1 & \Phi_2 \\ \Phi_2^* & \Phi_3 \end{bmatrix}
    =\begin{bmatrix} \bigoplus_{i=1}^n \Phi_{1,i} &  \bigoplus_{i=1}^n \Phi_{2,i} \\ \bigoplus_{i=1}^n\Phi_{2,i}^* & \bigoplus_{i=1}^n \Phi_{3,i} \end{bmatrix}
\end{align*}
is such that \eqref{eq:DeltaIQC} holds with $u=\alpha\Delta(y) = (\bigoplus_{i=1}^n \alpha\Delta_i)(y)$, and $y=(y_1,\ldots,y_n)$. As such, by Theorem~\ref{thm:robust_stability}, 
if there exists $\epsilon\in\mathbb{R}_{>0} $ such that for all $u \in \L_{2\,}^{2m}$,
 \begin{equation}
 \label{eq:stability_l2+}
    	\innerp{\begin{bmatrix}
	  y \\
		u 
        \end{bmatrix}}{\begin{bmatrix} \Phi_1 & \Phi_2 \\ \Phi_2^* & \Phi_3 \end{bmatrix}\begin{bmatrix}
	  y \\
		u 
        \end{bmatrix}} \leq -\epsilon \| u \|_2^2, 
\end{equation}
with $y=Gu$, then $\sys{G}{\Delta}$ is stable, which implies uncertain network stability by Theorem~\ref{thm:stable_sampledH}. As noted in Remark~\ref{rem:coprime}, the potential lack of structure in $G$ translates to a lack of apparent structure in the IQC robust stability certificate \eqref{eq:stability_l2+}. However, given the structured stable coprime factors $N$ and $M$ in \eqref{eq:coprimeNM}, in view of Remark~\ref{rem:coprime}, the existence of $\epsilon\in\mathbb{R}_{>0}$ such that \eqref{eq:stability_l2+} for all $u\in\L_{2\,}^{2m}$, is equivalent to the existence of $\epsilon\in\mathbb{R}_{>0}$ such that for all $z\in\L_{2\,}^{2m}$,
\begin{align}
\label{eq:stabiliy_G_2}
        \innerp{\begin{bmatrix}
	  N \\
		M 
        \end{bmatrix}z}{\begin{bmatrix} \Phi_1 & \Phi_2 \\ \Phi_2^* & \Phi_3 \end{bmatrix}  
        \begin{bmatrix}
	  N \\
		M 
        \end{bmatrix}z}
         \leq -\epsilon \| z \|_2^2,
\end{align} 
since $M$ and $M^{-1}$ are both stable systems, whereby $-\epsilon \|z\|_2^2 \leq -\epsilon \|Mz\|_2^2/\|M\|^2$ and $-\epsilon \|Mz\|_2^2 \leq - \epsilon \|z\|_2^2/\|M^{-1}\|^2$; for convenience, $M$ (resp.~$M^{-1}$) is not distinguished from the restriction to $\L_{2\,}^{2m}$, and likewise for all stable systems henceforth. 

While \eqref{eq:stabiliy_G_2} is structured, the IQC is not directly amenable to decomposition via the approach underlying~\cite[Theorem 1]{pates2016scalable}. The network structure here is contained in the stable coprime factor $M = F - L$, where $L = \sum_{k= 1}^m L_k$ is the sub-system graph Laplacian matrix in \eqref{eq:Laplacian_Subsystem_Graph},  $L_k=B_{(\cdot,k)}B_{(\cdot,k)}^\prime$ for the 
corresponding incidence matrix $B$, and $F:= I_{2m}-T\circ H$, with block-diagonal $T$ and $H$  as per~\eqref{eq:blkdiag}. The other coprime factor $N=H$ is also block diagonal. It follows that \eqref{eq:stabiliy_G_2} is equivalent to
\begin{align}
  \innerp{\begin{bmatrix}
	  I_{2m}\\
		L \\
        \end{bmatrix}z}{\left[ \begin{array}{cc} \Xi_1 & \Xi_2 \\
    	\Xi_2^{*}& \Xi_3\end{array} \right] 
        \begin{bmatrix}
	  I_{2m}\\
		L \\
        \end{bmatrix}z} \leq -\epsilon \| z \|_2^2
       \label{eq:stabiliy_G_3b}
\end{align}
 for $z \in \L_{2\,}^{2m}$, where
\begin{subequations}
    \label{t4Pi}
    \begin{align}
    \Xi_{1}&:= N^*\Phi_{1}N+ N^*\Phi_2F+F^*\Phi_2^*N+F^*\Phi_3 F~,\\
    \Xi_{2}&:=-N^*\Phi_2-F^*\Phi_3~,\\
    \Xi_{3}&:=\Phi_3~.
\end{align}
\end{subequations}
In view of the block-diagonal structure of $N$, $F$, $\Phi_1$, $\Phi_2$, $\Phi_3$, and the Hilbert adjoints when restricted to $\L_{2\,}$, $\Xi_1 =  \bigoplus_{i=1}^n \Xi_{1,i}$, $\Xi_{2}= \bigoplus_{i=1}^n \Xi_{2,i}$, and $\Xi_{3} =\bigoplus_{i=1}^n \Xi_{3,i}$, where 
\begin{align*}
    \Xi_{1,i} &:= H_i^*\Phi_{1,i}H_i +(I_{m_i} - \mathbf{1}_{m_i}H_i)^*\Phi_{3,i} (I_{m_i} - \mathbf{1}_{m_i}H_i) \\
    &\quad + H_i^*\Phi_{2,i}(I_{m_i} - \mathbf{1}_{m_i}H_i) +(I_{m_i} - \mathbf{1}_{m_i}H_i)^*\Phi_{2,i}^*H_i~,\\
    \Xi_{2,i}&:=-H_i^*\Phi_{2,i}-(I_{m_i}-\mathbf{1}_{m_i}H_i)^*\Phi_{3,i}~,\\
    \Xi_{3,i}&:=\Phi_{3,i}~.
\end{align*}

\begin{remark}
    While the structure of \eqref{eq:stabiliy_G_3b} resembles that of the relevant IQC in the proof of~\cite[Theorem 1]{pates2016scalable} (see eq.~(6) there), a key difference exists. In particular, $z$ on the right of \eqref{eq:stabiliy_G_3b} is not the sum of $L_k z$ over $k\in[1:m]$. Further, working with a version of Theorem~\ref{thm:robust_stability} that has $\|y\|_2^2$ on the right-hand side of \eqref{eq:stab_G} (as exploited to arrive at (6) in~\cite{pates2016scalable}) instead of $\|u\|_2^2$, does not resolve the issue, since this only leads to $y=Nz$ on the right, which does not involve $L$. As such, it is still not possible to directly apply the setup in~\cite{pates2016scalable}.
\end{remark} 

The following expands upon ideas from the proof of~\cite[Theorem 1]{pates2016scalable} to enable decentralized verification of \eqref{eq:stabiliy_G_3b}:
\begin{lemma}\label{PatesReadapted}
Let $W=\sum_{k=1}^m W_k \in \mathbb{R}^{2m\times 2m}$ be such that $W_k\succeq 0$ and $W\succ 0$.
Suppose there exist $X_k\!=\!X^*_k:\L_{2\,}^{2m}\rightarrow\L_{2\,}^{2m}$,  $Z_k\!=\!Z_k^*:\L_{2\,}^{2m}\rightarrow\L_{2\,}^{2m}$, and $\epsilon_k>0$, $k\in[1:m]$, such that for all $z\in\L_{2\,}^{2m}$,
\begin{align}
    \label{ISCk}
       & \innerp{\begin{bmatrix}
		I_{2m} \\
		L_k\\
		\end{bmatrix} z} {\begin{bmatrix}
		X_k+\epsilon_k W_k & \Xi_{2} \\
    	\Xi_{2}^{*} & Z_k \end{bmatrix}
     \begin{bmatrix}
		I_{2m} \\
		L_k\\
		\end{bmatrix}z}\leq 0,~ k\in[1:m],\\
 \label{ISCCompa}
    &       \innerp{z}{\Xi_1 z} \leq \sum_{k=1}^m \innerp{z}{X_k z},\\
\label{ISCCompb}
      &     \innerp{L z}{\Xi_3 L z} \leq \sum_{k=1}^m\innerp{L_k z}{Z_k L_k z},
\end{align}   
where $L=\sum_{k=1}^m L_k$ is the sub-system graph Laplacian in~\eqref{eq:Laplacian_Subsystem_Graph}.
 Then, there exists $\epsilon\in\mathbb{R}_{>0}$ such that \eqref{eq:stabiliy_G_3b} for all $z\in\L_{2\,}^{2m}$.
  %
 \end{lemma}
    \begin{proof}
	 For all $z\in\L_{2\,}^{2m}$, and $k\in[1:m]$, \eqref{ISCk} implies
	    \begin{align*}
	        \innerp{z}{X_k z} \!+\! \innerp{L_k z}{\Xi_3 L_k z} \!+\! \innerp{z}{\Xi_2 L_k z} \!+\! \innerp{L_k z}{\Xi_2^* z} \!  \leq\! -\epsilon_k\! \innerp{z}{W_k z}\!, 
	    \end{align*}
      and therefore, 
	    \begin{align}
	        \nonumber 
         \sum_{k=1}^m  ( \innerp{z}{X_k z} + \innerp{z}{\Xi_2 L_k z}  + \innerp{L_k z}{\Xi_2^* z} + \innerp{L_k z}{\Xi_3 L_k z}  )\\
         \leq -\sum_{k=1}^m \epsilon_k \innerp{z}{W_k z} \leq
          -\epsilon \| z \|_2^2, \label{ProofISC2}
	    \end{align}
      where $\epsilon = \big(\min_{k\in[1:m]} \epsilon_k\big) \cdot \big(\min_{x\in\mathbb{R}^{n}} x^\prime W x/x^\prime x\big)>0$; note that 
      $\sum_{k=1}^m \epsilon_k W_k \succeq \big(\min_{k\in[1:m]} \epsilon_k\big) \sum_{k=1}^m W_k$, because each $W_k\succeq 0$.
	    Combining \eqref{ISCCompa}, \eqref{ISCCompb}, and \eqref{ProofISC2}, gives 
     \begin{equation*}
     \innerp{z}{\Xi_1 z} + \innerp{z}{\Xi_2 L z}  + \innerp{L z}{\Xi_2^*  z} + \innerp{L z}{\Xi_3 L  z}  \leq -\epsilon \| z \|_2^2,
     \end{equation*}
     which is \eqref{eq:stabiliy_G_3b} as claimed.
	\end{proof}

\section{Main result: Decentralized conditions}
\label{sec:struct}

In this section it is shown how the ingredients $W_k$, $X_k$, and $Z_k$, $k\in[1:m]$, in Lemma~\ref{PatesReadapted} may be selected to devise a collection of scalable conditions that together verify \eqref{eq:stabiliy_G_3b} for all $z\in\L_{2\,}^{2m}$,
and therefore, stability of $\sys{G}{\Delta}$, from which robust network stability follows by Theorem~\ref{thm:stable_sampledH}. The conditions are decentralized in the sense that each depends on agent model data that is local to a corresponding link, and a relevant component of the IQC based uncertainty description. To this end, a potentially conservative restriction is used to facilitate decentralized verification of \eqref{ISCCompb} in particular. This restriction corresponds to use of the decomposition $\Xi_{3,i} = E_{i} + D_{i}$, with 
$E_{i} = E_{i}^*:\L_{2\,}^{m_i}\rightarrow\L_{2\,}^{m_i}$ negative semi-definite, and $D_{i}=(\oplus_{k=1}^{m_i} D_{i,k}):\L_{2\,}^{m_i}\rightarrow\L_{2\,}^{m_i}$ diagonal. That such decomposition is always possible, although not uniquely, is formalized in the following lemma (if $\Xi_{3,i}=\Phi_{3,i}$ is already diagonal, then taking $E_{i}=O_{m_i}$ incurs no conservativeness):
\begin{lemma}
\label{lem:decomp}
Given bounded linear $\Xi=\Xi^*:\L_{2\,}^{p} \rightarrow \L_{2\,}^{p}$, $p\in\naturals$, there exist negative semi-definite $E=E^*:\L_{2\,}^{p} \rightarrow \L_{2\,}^{p}$, and diagonal $D:\L_{2\,}^{p} \rightarrow \L_{2\,}^{p}$, such that $\Xi=E+D$. 
\end{lemma}
\begin{proof}
In view of Lemma 2.7-2 and  Theorem 9.2-1 in~\cite{Kreyszig}, there exist $\lambda_1,\lambda_2 \in \real$ such that 
$\lambda_1 \|x\|_2^2 \leq \innerp{x}{\Xi x} \leq \lambda_2 \|x\|_2^2$ for all $x\in\L_{2\,}^p$. With $D:=\lambda_2  \cdot I_p$ and $E:=\Xi - \lambda_2  \cdot I_p$, $\Xi=E+D$ with $D:\L_{2\,}^p \rightarrow \L_{2\,}^p$ diagonal, and $E:\L_{2\,}^p \rightarrow \L_{2\,}^p$ negative semi-definite, because $\langle x, E x\rangle =\langle x, \Xi x\rangle + \langle x,-\lambda_2 x\rangle \leq  \lambda_2 \|x\|_2^2 - \lambda_2 \langle x, x\rangle=0$ for every $x\in \L_{2\,}^p$.  
\end{proof}
%

The subsequent matrix definition, and related properties, lead to the proposed selection of $W_k$, $X_k$, and $Z_k$, in Lemma~\ref{PatesReadapted} for decentralized verification of \eqref{ISCk}, \eqref{ISCCompa}, and \eqref{ISCCompb}. The definition pertains to the structure of the networked system m model, encoded by the network graph $\G=(\V,\E)$, where $\V=[1:n]$, $|\E|=m$, $|\N_i|=m_i$, and sub-system graph $\widetilde{\G}=(\widetilde{\V},\widetilde{\E})$, where $|\widetilde{\V}|=2m$, and $|\widetilde{\E}|=m$; see Section~\ref{sec:ideal}. In particular, for each $k\in[1:m]$, define 
\begin{align} \label{eq:whatB}
    \widehat{B}_k := (\mathrm{diag}(B_{(\cdot,k)}))^2 \in \mathbb{R}^{2m\times 2m}, 
\end{align}
where $B\in\mathbb{R}^{2m\times 2m}$ is the incidence matrix of the sub-system graph Laplacian matrix $L=\sum_{k=1}^m L_k = \sum_{k=1}^m B_{(\cdot,k)}B_{(\cdot,k)}^\prime$ in \eqref{eq:Laplacian_Subsystem_Graph}. The matrix $\widehat{B}_k$ is a diagonal matrix, with entries in $\{0,1\}$, and diagonal sparsity pattern that aligns with the two uncertain links between the agents associated with $r\in\widetilde{\V}$ and $s\in\widetilde{\V}$ for the given $k\in[1:m]$; i.e., agents $\{i,j\}=\kappa_{\E}^{-1}(k)$, where $\{r,s\}=\kappa_{\widetilde{\E}}^{-1}(k)$, assuming compatible enumerations, as specified where the incidence matrix $B$ is defined below \eqref{eq:Laplacian_Subsystem_Graph}. It can be shown by direct calculation that $\widehat{B}_k B_{(\cdot,k)}=(\mathrm{diag}(B_{(\cdot,k)}))^2 B_{(\cdot,k)}=B_{(\cdot,k)}$,
\begin{subequations}
\label{eq:BhatMagic}
\begin{align}
&\widehat{B}_k L_k = \widehat{B}_k B_{(\cdot,k)} B_{(\cdot,k)}^\prime = B_{(\cdot,k)} B_{(\cdot,k)}^\prime = L_k,
    \label{eq:BhatMagicC} \\
&\widehat{B}_k^\prime=\widehat{B}_k,~~~ \widehat{B}_k\widehat{B}_k=\widehat{B}_k, ~\text{ and }~  \widehat{B}_k~\mathrm{diag}{(d)}~\widehat{B}_l = O_{2m}, 
\label{eq:BhatMagicD}
\end{align}
for all $k\neq l\in[1:m]$, and $d\in\mathbb{R}^{2m}$, whereby
\begin{align} \label{eq:BhatMagicA}
    &(\sum_{k=1}^{m} \widehat{B}_k) ~(\sum_{l=1}^{m} \widehat{B}_l)=\sum_{k=1}^m\sum_{l=1}^m \widehat{B}_k\widehat{B}_l= \sum_{k=1}^m \widehat{B}_k = I_{2m}~,\\
    \begin{split}
    &(\sum_{k=1}^{m} \widehat{B}_k) ~\mathrm{diag}(d) ~(\sum_{l=1}^{m} \widehat{B}_l)  \\
    &  \quad = 
    \sum_{k=1}^m\sum_{l=1}^m \widehat{B}_k ~ \mathrm{diag}(d) ~ \widehat{B}_l 
    = \sum_{k=1}^{m} \widehat{B}_k~\mathrm{diag}(d)~\widehat{B}_k~.
    \end{split}\label{eq:BhatMagicB}
\end{align}
\end{subequations}
\begin{proposition}\label{prop:sharing_couples}
For each $k\in[1:m]$, let 
\begin{align}
    W_k &:= \widehat{B}_k~, \label{eq:t3P0}\\
    X_k &:= \frac{1}{2}\big(\widehat{B}_k\Xi_1 + \Xi_1 \widehat{B}_k\big)~, \label{eq:t3P1}\\
   Y_k &:= \Xi_2\widehat{B}_k~, \label{eq:t3P2}\\
    Z_k &:= 
    \widehat{B}_k\big(\bigoplus_{i=1}^n D_i\big)\widehat{B}_k~, \label{eq:t3P3}
\end{align}
where $\widehat{B}_k$ is defined in \eqref{eq:whatB}, and diagonal $D_i:\L_{2\,}^{m_i}\rightarrow\L_{2\,}^{m_i}$ is such that $\Xi_{3,i}=D_i+E_i$ for some negative semi-definite $E_i:\L_{2\,}^{m_i}\rightarrow\L_{2\,}^{m_i}$, $i\in[1:n]$, with $\Xi_1$, $\Xi_2$, $\Xi_3$ as per \eqref{t4Pi}. Given this,
if for all $k\in[1:m]$, there exists $\epsilon_k\in\mathbb{R}_{>0}$ such that
  \begin{equation}\label{PCSk2} 
        \innerp{\begin{bmatrix}
		I_{2m} \\
		L_k\\
		\end{bmatrix}z}{\left[ \begin{array}{cc} X_k+\epsilon_k W_k & Y_k \\
    	Y^{*}_k & Z_k \end{array} \right]\begin{bmatrix}
		I_{2m} \\
		L_k\\
		\end{bmatrix} z}\leq0
    \end{equation}
    for all $z\in\L_{2\,}^{2m}$,
 then \eqref{eq:stabiliy_G_3b} for all $z\in\L_{2\,}^{2m}$.
\end{proposition}
\begin{proof}
With $Z_k$ as per \eqref{eq:t3P3}, 
in view of \eqref{eq:BhatMagicB}, and the hypothesis $\innerp{z}{(\oplus_{i=1}^n E_i) z} \leq 0$ for all $z\in\L_{2\,}^{2m}$,
\begin{align*}
\sum_{k=1}^m \innerp{L_k z}{Z_k L_k z}
&=\!
\sum_{k=1}^m  \innerp{\widehat{B}_k L_k z}{(\oplus_{i=1}^n D_i) \widehat{B}_k L_k z}\\
&\geq\! \sum_{k=1}^m \innerp{L_k z}{(\oplus_{i=1}^n D_i) L_k z} \!+\!
\innerp{L z}{(\oplus_{i=1}^n E_i) L z} 
\nonumber \\
&=\!  \innerp{L z}{\Xi_3 L z},
\end{align*}
which is \eqref{ISCCompb}. Similarly, with $X_k$ as per \eqref{eq:t3P1}, it follows that 
\begin{align*}
    \sum_{k=1}^m X_k
    = \frac{1}{2}\Big(\big( \sum_{k=1}^m \widehat{B}_k \big)\Xi_1 + \Xi_1 \big( \sum_{k=1}^m \widehat{B}_k\big)\Big) = \Xi_1,
\end{align*}
so that for all $z\in\L_{2\,}^{2m}$, $\sum_{k=1}^m \innerp{z}{X_k z} \leq \innerp{z}{\Xi_1 z}$, which is \eqref{ISCCompa}. Also, with 
$Y_k$ as per \eqref{eq:t3P2}, $Y_k L_k = \Xi_2 \hat{B}_k L_k =  \Xi_2 L_k$ in view of \eqref{eq:BhatMagic}, whereby \eqref{PCSk2} is the same as \eqref{ISCk} for all $k\in[1:m]$, and $z\in\L_{2\,}^{2m}$. Finally, with
$W_k\succeq 0$ as per \eqref{eq:t3P0}, $\sum_{k=1}^{m} W_k = \sum_{k=1}^{m} \widehat{B}_k = I_{2m} \succ 0$; see~\eqref{eq:BhatMagicA}. As such, Lemma~\ref{PatesReadapted} applies, and therefore, \eqref{eq:stabiliy_G_3b} for all $z\in\L_{2\,}^{2m}$ as claimed.
\end{proof}

\begin{remark}
While, for any given $k\in[1:m]$, the IQC in~\eqref{PCSk2} is expressed in terms of $z=(z_1,\ldots,z_{2m})\in\L_{2\,}^{2m}$, only $m_i+m_j$ of the co-ordinates contribute to the left side of the inequality, where $\{i,j\}= \kappa_{\E}(k)$. These co-ordinates pertain to the particular pair of links associated with $k$, being those transformed by corresponding components of the model data local to agents $i\in\V$ and $j\in\V$, and the structured link uncertainty description, embedded in $\Xi_1$, $\Xi_2$, $\Xi_3$ as per~\eqref{t4Pi}. It is in this sense that the collection of conditions \eqref{PCSk2}, over $k\in[1:m]$, is a decentralized certificate for \eqref{eq:stabiliy_G_3b}, and therefore, networked system robust stability with respect to the specified link uncertainty. The proposed decomposition is just one possibility. Alternatives are under investigation.
\end{remark}

\section{Conclusions}
\label{sec:conc}
Link-wise decentralized robust stability conditions are devised for networked systems in the presence of link uncertainty. This result is based on input-output IQCs that are used to describe the link uncertainties, and ultimately the corresponding structured robust stability certificate. Future work will explore alternative decompositions with a view to reducing conservativeness subject to maintaining scalability. It is also of interest to apply the main result to study specific networked system scenarios where information exchange is impacted by asynchronous time-varying delays, and dynamic quantization, for example.

\appendix
\noindent\textbf{\emph{Proof of Lemma~\ref{lem:stabstab}.}} \ref{it:stable}$)\implies$\ref{it:causalinv}). Consider $d_u=0$, and arbitrary $d_y=d\in\L_{2e}^p$ in \eqref{eq:feedback1}. By \ref{it:stable}), there exists unique $(y,u)\in\L_{2e}^p\times \L_{2e}^q$ such that $y=(G\circ\Delta)(y) + d$, $u=\Delta(y)$. In particular, $(I_p-G\circ\Delta) = (y \mapsto d)$ is a bijection on $\L_{2e}^p$, and the inverse $(d \mapsto y)$ is causal as $[\![G,\Delta]\!]$ is well-posed. Further, there exists $c_0>0$ such that $\|y\|_2 \leq \|(y,u)\|_2 \leq c_0 \cdot \|(d_y,0)\|_2 = c_0 \cdot \|d\|_2$ when $d_y=d\in\L_{2\,}^p$, since $[\![G,\Delta]\!]=((d_y,d_u)\mapsto(y,u))$ is stable.

\ref{it:causalinv})$\implies$\ref{it:stable}). By \ref{it:causalinv}), for every $d\in\L_{2e}^{2m}$, there exists unique $v\in\L_{2e}$ such that the following hold: $v = (G\circ \Delta)(v) + d$; the (inverse) map $d\mapsto y$ is causal; and there exists $c_0>0$ such that $\|y\|_2 \leq c_0 \cdot \|d\|_2$ whenever $d\in\L_{2\,}^{2m}$. Since $G$ is linear, it follows that for every $(d_y,d_u)\in\L_{2e}^p\times\L_{2e}^{q}$, there exists unique $y\in\L_{2e}^p$ such that $y = (G\circ\Delta)(y) + (G d_u + d_y) = G u + d_y$, where $u = \Delta(y) + d_u$; c.f.,~\eqref{eq:feedback1}. Moreover, since $G$ is causal, the composition of maps $(d_y,d_u) \mapsto (d_y+G d_u) \mapsto y$ is causal, and $\|y\|_2 \leq c_0\cdot \|G d_u + d_y\|_2 \leq c_0\cdot \|d_y\|_2 + c_0\cdot\|G\|\cdot\|d_u\|_2 \leq c_1 \cdot \|(d_y,d_u)\|_2$, with $c_1^2 \geq c_0^2(1+\|G\|^2)$, whenever $(d_y,d_u)\in\L_{2\,}^p\times\L_{2\,}^{q}$. Since $\Delta$ is stable, it follows that the composition $(d_y,d_u) \mapsto (y,d_u) \mapsto u$ is also causal, and $\|u\|_2 \leq \|\Delta(y)\|_2 + \|d_u\|_2 \leq \|\Delta\|\cdot c_0 \cdot \|d_y\|_2 + (\|\Delta\| \cdot c_0 \cdot \|G\| + 1) \|d_u\|_2 \leq c_3\cdot \| (d_y,d_u) \|_2$ with $c_3^2 \geq \|\Delta\|^2\cdot c_0^2 + (\|\Delta\| \cdot c_0 \cdot \|G\| + 1)^2$. Therefore, statement \ref{it:stable}) holds. \hfill $\square$

\bibliography{Bib}
\bibliographystyle{ieeetr}
\end{document}